\documentclass{rsproca_new}
\usepackage{graphicx}
\usepackage{enumitem}
\usepackage{amssymb,amsthm,amsmath}
\usepackage{url}
\usepackage[numbers]{natbib}

\usepackage[charter,cal=cmcal,sfscaled=false]{mathdesign}

\usepackage{hyperref}

\hypersetup{
  pdfstartview=Fit,
  pdfpagelayout=SinglePage,
  colorlinks,
  citecolor=blue,
  linkcolor=blue,
  urlcolor=blue
}
\usepackage[figure]{hypcap}

\def\Z{{\mathbb Z}}

\def\P{{\cal P}}

\def\R{{\mathbb R}}

\def\T{{\cal T}}

\def\ENERGY{\mbox{\textbf{ENERGY}}}
\def\MARGINAL{\mbox{\textbf{MARGINAL}}}

\def\be{\begin{equation}}
\def\ee{\end{equation}}
\newtheorem{theo}{Theorem}
\newtheorem{prop}[theo]{Proposition}
\newtheorem{lemma}{Lemma}

\newtheorem{defin}{Definition}

\begin{document}

\title{Two Dimensional Translation-Invariant Probability Distributions: \\Approximations, Characterizations and No-Go Theorems}
\author{Zizhu Wang$^{2,1}$, Miguel Navascu\'es$^{1}$}
\address{$^{1}$ Institute for Quantum Optics and Quantum Information (IQOQI) Vienna, \\Austrian Academy of Sciences,\\ Boltzmanngasse 3, 1090 Vienna, Austria\\
$^{2}$ Institute of Fundamental and Frontier Sciences, \\University of Electronic Science and Technology of China,\\ Chengdu 610054, China
}
\corres{Zizhu Wang and Miguel Navascu\'es\\
\email{\{zizhu.wang,miguel.navascues\}@oeaw.ac.at}}

\begin{abstract}
We study the properties of the set of marginal distributions of infinite translation-invariant systems in the 2D square lattice. In cases where the local variables can only take a small number $d$ of possible values, we completely solve the marginal or membership problem for nearest-neighbors distributions ($d=2,3$) and nearest and next-to-nearest neighbors distributions ($d=2$). Remarkably, all these sets form convex polytopes in probability space. This allows us to devise an algorithm to compute the minimum energy per site of any TI Hamiltonian in these scenarios exactly. We also devise a simple algorithm to approximate the minimum energy per site up to arbitrary accuracy for the cases not covered above. For variables of a higher (but finite) dimensionality, we prove two no-go results. To begin, the exact computation of the energy per site of arbitrary TI Hamiltonians with only nearest-neighbor interactions is an undecidable problem. In addition, in scenarios with $d\geq 2947$, the boundary of the set of nearest-neighbor marginal distributions contains both flat and smoothly curved surfaces and the set itself is not semi-algebraic. This implies, in particular, that it cannot be characterized via semidefinite programming, even if we allow the input of the program to include polynomials of nearest-neighbor probabilities.
\end{abstract}

\maketitle

\section{Introduction}
The distribution of stars at large scales, the vacuum state of a quantum field theory, the thermal state of any solvable spin model: these are examples of systems with infinitely many sites where the description of a bounded environment does not depend on its location within the lattice. We call this property translation invariance (TI). 

An agent exploring an infinite translation-invariant world would find that the statistics of the local random variables which he can access are constrained by the requirement of infinite TI. However, and despite a long history of research on TI systems, driven by the needs of statistical physics (see, e.g.~\cite{baxter}), it is far from clear what those constraints exactly are.

This conundrum is at the essence of the \emph{TI $\MARGINAL$ problem}, where a number of probability distributions of finitely-many variables are provided and the task is to certify if they correspond to the marginals of a TI system. $\MARGINAL$ arises naturally at the intersection of quantum information science and condensed matter physics, when we try to determine whether the dynamical structure factors of a large spin system are compatible with an underlying Bell local quantum state~\cite{ours}.

$\ENERGY$, the dual problem of computing the minimum energy per site of a translation-invariant Hamiltonian, is also of special concern for the mathematical physics community. While $\ENERGY$ is efficiently solvable in one dimensional ($1D$) systems~\cite{Schlijper1985,Pivato,Goldstein2017}, in $2D$ very few solvable instances are documented.

What do we know about $\MARGINAL$ or $\ENERGY$ in $2D$? Not much. We know that the set of TI marginals is closed and convex~\cite{chazottes,Goldstein2017}. We also know that that the problem of determining the existence of near-neighbor marginals of a TI distribution with a given (finite) support is undecidable~\cite{ruelle1978,SchlijperTiling,Pivato}. In ~\cite{kuna2011} the case where the random variables are dichotomic is studied and necessary and sufficient conditions are provided for the existence of a TI extension, given the \emph{infinitely many} marginal distributions of any $n$ lattice sites. In~\cite{Goldstein2017} explicit examples of probability distributions for lattice sites within a $2\times 2$ square are given which, despite satisfying all local symmetries associated to TI, do not admit $2D$ TI extensions. Similar instances of $2D$ non-extendible distributions constructed from $3\times 3$ squares satisfying additional rotation and reflection symmetries are also provided. In~\cite{chazottes} the approximability of the $\MARGINAL$ problem via convex polytopes is considered.

It is still open under which scenarios one can solve $\MARGINAL$ and $\ENERGY$ exactly, or how to attack these problems numerically~\footnote{Brute-force computational methods to decide Bell nonlocality in the simplest $2D$ TI scenarios or to compute lower bounds on ground state energy densities soon become intractable due to lack of computer memory.}. Even the geometry of the set of all marginal distributions of $2D$ TI systems remains a mystery: in $1D$, the marginal distributions of the variables of finitely many sites of TI systems form a convex polytope in the space of probabilities~\footnote{A convex polytope is a convex set defined by a finite number of linear inequalities~\cite{grunbaum2003convex,ziegler1995lectures}.}. Should we expect this to hold in $2D$ systems as well?

Motivated by the desire of understanding the nature of quantum nonlocality in $2D$ materials \cite{ours}, in this paper we will considerably advance the above fundamental questions. We will show that, in scenarios where the random variables take a small number $d$ of possible values, the $\MARGINAL$ problem for nearest-neighbor (for $d=2,3$) and next-to-nearest neighbor distributions (for $d=2$) is exactly solvable. We also prove solvable a natural variant of $\MARGINAL$ where the TI extension is also required to satisfy invariance under reflection of the horizontal and vertical axes. An immediate consequence of these results is that, in any of the above cases, it is possible to solve $\ENERGY$ exactly. For the scenarios not covered above, we provide an algorithm to approximate the set of TI marginals or solve the $\ENERGY$ problem up to arbitrary precision.

On the negative side, we show that, for arbitrary TI Hamiltonians, $\ENERGY$ is undecidable. We also prove that, contrary to the solvable cases, for local random variables of high cardinality, the set of nearest-neighbor distributions admitting a TI extension is no longer a convex polytope in the space of probabilities or even a semi-algebraic set\footnote{A semi-algebraic set is the union of finitely many regions of $\R^n$ defined by a finite number of polynomial inequalities~\cite{Bochnak1998}.}. This implies, in particular, that standard tools from convex optimization, such as linear programming \cite{linear} or semidefinite programming \cite{sdp} cannot be used to fully characterize Bell nonlocality in large $2D$ condensed matter systems.

The structure of this paper is as follows: in Section \ref{basic_prop} we will define and list known properties of the set of marginal 2D TI distributions. We will introduce the problems $\MARGINAL$ and $\ENERGY$, and show how to solve their 1D versions. We will also describe a simple symmetrization process that will play an important role in the mathematical proofs to come. The remaining sections present our original contributions. In Section \ref{properties} we will present a practical algorithm to characterize the set of 2D marginals up to arbitrary accuracy. Later, in Section \ref{positive}, we will provide a few instances of the marginal problem which are exactly solvable. In Section \ref{negative} we will prove two no-go theorems about the set of 2D TI marginals, namely: (a) the problem of computing $\ENERGY$ exactly is undecidable; (b) for random variables with support $d$ greater than or equal to $2947$, the set of nearest-neighbor distributions is not a semi-algebraic set.

\section{The set of marginal 2D TI distributions: definition and known properties}
\label{basic_prop}
Consider an infinite two-dimensional square lattice, and suppose that each site $(x,y)\in\Z^2$ has access to a local random variable $a_{x,y}$ that can take $d<\infty$ possible values. We will call $d$ the \emph{local dimension} of the lattice. For any finite subset $K$ of $\Z^2$, we will assume that the variables $a_K\equiv\{a_z:z\in K\}$ follow a probability distribution $P_K(a_K)$. We say that the system is translation-invariant (TI) if, for any finite set $K$ and any vector $c\in\Z^2$, 

\be
P_{K}=P_{K+c},
\label{TI_def}
\ee
\noindent where $K+c=\{z: z-c\in K\}$.

TI is a very common property in nature. It is satisfied approximately in large crystal structures, and exactly by any stationary state in quantum field theory. In addition, most exactly solvable models in statistical physics comply with this symmetry in the thermodynamic limit of infinitely many sites.

In this paper we will consider a scenario where an agent conducts local observations within an infinite 2D TI system. Intuitively, the condition of TI means that the statistics of the variables within the immediate vicinity of the agent do not give any clue about his position in the lattice. Our goal is to determine how the statistics of such a sample are constrained by the requirement that these variables arise from a 2D TI system.

\begin{figure}
  \centering
  \includegraphics[width=0.6\textwidth]{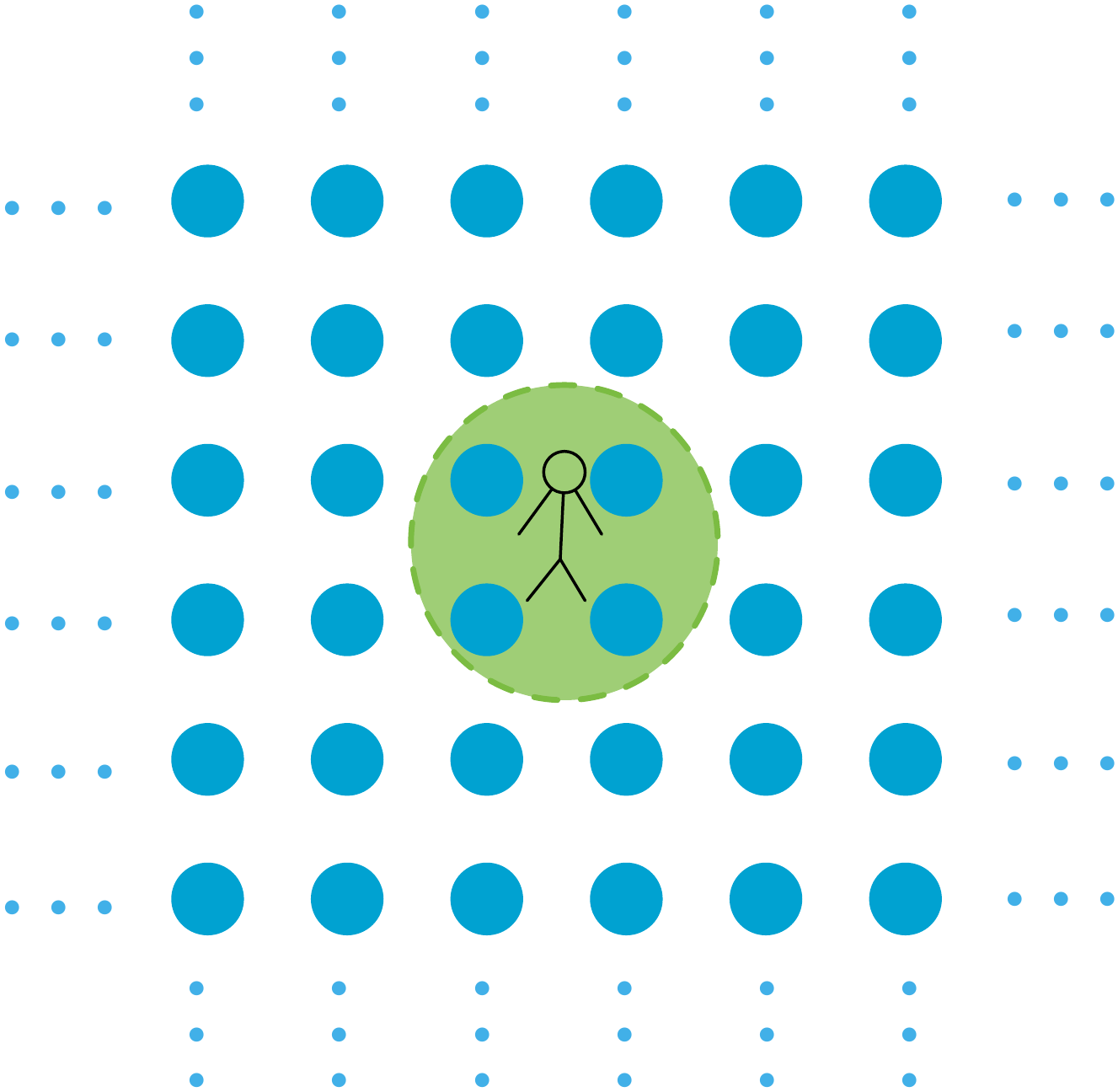}
  \caption{\textbf{An agent exploring his surroundings in an infinite TI lattice.} In the picture, the agent only has access to the random variables corresponding to four nearby sites. How are the statistics of the corresponding four random variables restricted by the TI condition?}
  \label{little_man}
\end{figure}

With a slight abuse of notation, any function $P$ assigning probability distributions $P_K(a_K)$ to any finite set $K\subset \Z^2$ satisfying eq. (\ref{TI_def}) and the consistency conditions $\sum_{a_{L}}P_{K\cup L}(a_K,a_{L})=P_K(a_K)$ for $K\cap L=\emptyset$ will be called a \emph{2D TI probability distribution}. Each of the distributions $P_K(a_K)$ will be denoted the \emph{marginal} of a 2D TI distribution, or simply a TI marginal. A few special finite subsets of $\Z^2$ will appear frequently throughout this article, so we will need special names for them. $v,h,+,-$, will denote, respectively, the sets $\{(0,0),(0,-1)\}$, $\{(0,0),(1,0)\}$, $\{(0,0),(1,1)\}$, $\{(0,0),(1,-1)\}$. The bi-variate distributions $P_h(a,b)$, $P_v(a,b)$ ($P_+(a,b)$, $P_-(a,b)$) hence describe the statistics between nearest neighbors (next-to-nearest neighbors). We will also use $m\times n$ to describe the rectangle $\{(x,y): 0\leq x< m,0\leq y< n\}$; and $|m\times n|$, for $\{(x,y): |x|< m,|y|< n\}$. Finally, $1,...,n$ will denote the set $\cup_{k=1}^n\{(k,0)\}$.

We will now introduce the two problems which will concern us through the rest of the article.

\begin{defin}[The TI $\MARGINAL$ problem]

Given a finite number of finite subsets of $\Z^2$, i.e., $\{K_i\}_{i=1}^n$, $|\cup_iK_i|<\infty$, and the probability distributions $\{Q_i(a_{K_i})\}_i$, determine if there exists a 2D TI distribution $P$ such that $P_{K_i}(a_{K_i})=Q(a_{K_i})$ for $i=1,...,n$. In the affirmative case, the distribution $P$ will be called a \emph{TI extension} of $\{Q_i(a_{K_i})\}_i$.
\end{defin}

\begin{figure}
  \centering
  \includegraphics[width=0.6\textwidth]{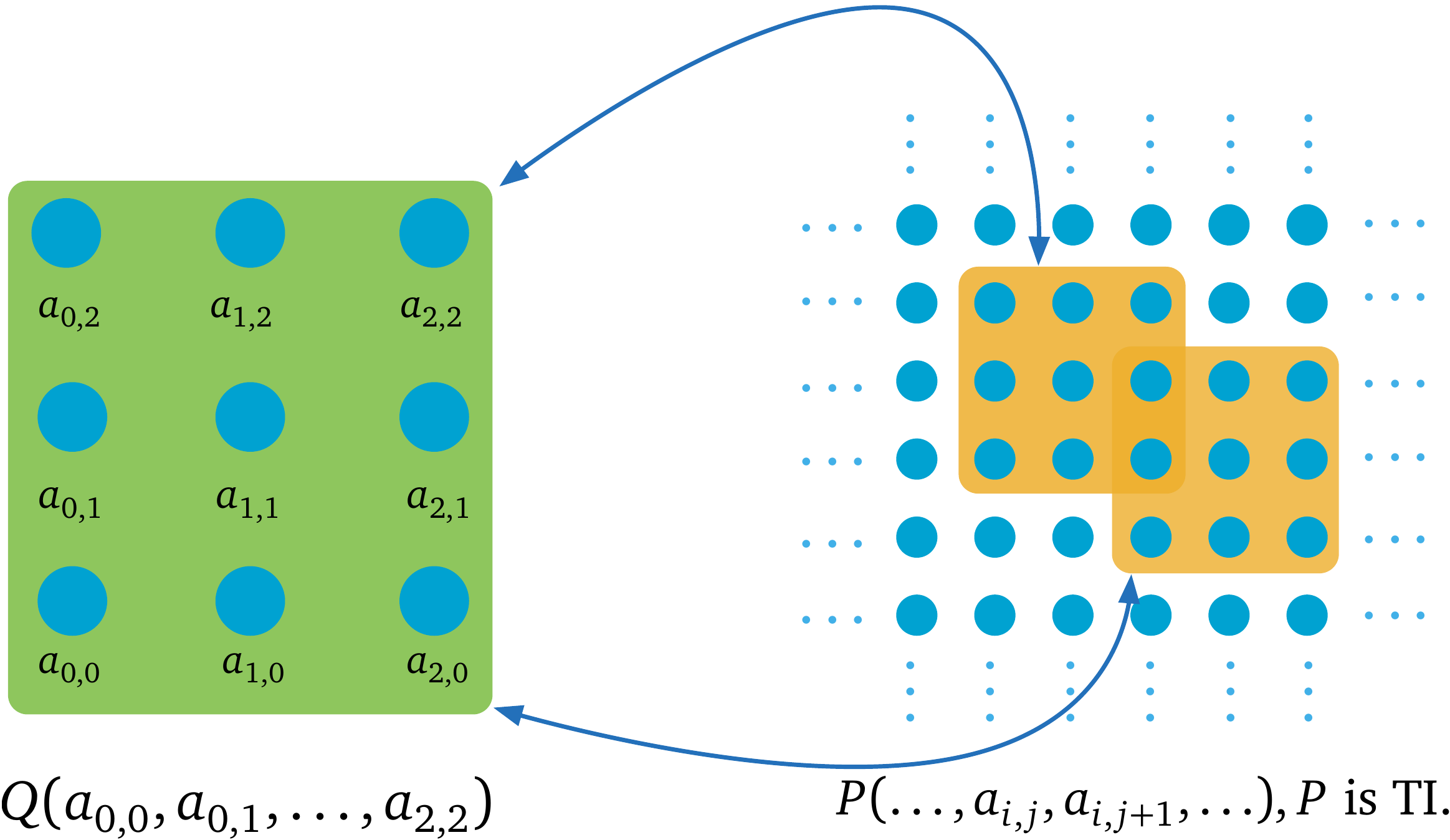}
  \caption{\textbf{The 2D TI marginal problem.} In the example, there is only one region, $K=3\times 3$. A probability distribution $Q_{3\times 3}$ for the $9$ variables $\{a_{i,j}:i,j=0,1,2\}$ is given (the green square), and the question is whether there exists a 2D TI system with a marginal probability distribution $P_{3\times 3}$ (the yellow squares) equal to $Q_{3\times 3}$.}
  \label{marginal}
\end{figure}

In the following, each configuration of cardinality $d<\infty$ and finite subsets of $\Z^2$ $\{K_i\}_{i=1}^n$ will be called \emph{a scenario}. In any scenario, the set of TI marginals is convex and closed \cite{Goldstein2017}.

Sometimes, e.g., in statistical physics, we are not interested in solving the marginal problem as much as in optimizing linear functionals over the set of marginals of 2D TI distributions. This motivates the following problem:

\begin{defin}[The $\ENERGY$ problem]

Given $\{K_i\}_{i=1}^n$, with $K_i\subset \Z^2$, $|\cup_iK_i|<\infty$, and functions $F_i:\{0,\ldots,d-1\}^{|K_i|}\to\R$, $i=1,...,n$, solve the optimization problem

\be
\min_{P}\sum_{i=1}^n\sum_{a_{K_i}}F_i(a_{K_i})P_{K_i}(a_{K_i}),
\label{energy_functional}
\ee
\noindent where the minimization is carried out over all TI distributions $P$.
\end{defin}

The $\ENERGY$ problem is the dual of the $\MARGINAL$ problem: given an oracle to solve one of them approximately, one can devise an algorithm that invokes the oracle a polynomial number of times in order to solve the other problem with a similar accuracy \cite{grotschel+}. Given $\{K_i\}_{i=1}^n$ and the interactions $\{F_i\}_{i=1}^n$, consider the family of (classical) 2D Hamiltonians of the form

\be
H_N(a_{N\times N})=\sum_{\{z: K_i+z\subset N\times N,\forall i\}}\sum_{i=1}^n\sum_{b=1}^{d^{|K_i|}}F_i(b)\delta(a_{K_i+z},b),
\ee

\noindent where each Hamiltonian $H_N$ describes the finite lattice $N\times N$. The \emph{minimum energy per site} of the infinite lattice is defined as $\lim_{N\to\infty} \frac{1}{N^2}\min_{a_{N\times N}}H_N(a_{N\times N})$. $\ENERGY$ takes its name from the observation that the minimum energy per site of the infinite 2D TI Hamiltonian defined by the terms $\{F_i\}_{i=1}^n$ equals $\ENERGY(\{F_i\})$. The reader can find a proof at the end of this section.

As we will see, the TI marginal problem is very hard. However, there exists an interesting variant whose solution turns out to be trivial.

\begin{defin}[The $\MARGINAL$ problem in 1D]

Let $Q_{1,...,s}(a_1,...,a_s)$ be a probability distribution. Determine if there exists a TI distribution $P$ such that $P_{1,...,s}(a_1,...,a_s)=Q_{1,...,s}(a_1,...,a_s)$.
\end{defin}

The solution of this problem is folklore among the community of condensed matter physicists: a distribution $Q_{1,...,s}(a_1,...,a_s)$ admits a TI extension iff $Q_{1,...,s-1}=Q_{2,...,s}$ \cite{Schlijper1985,Pivato,Goldstein2017}.

\vspace{10pt}

Given an arbitrary distribution $Q_{K}$ (not necessarily a TI marginal) of local variables over a large rectangle $K=m\times n$, one can derive a 2D TI distribution $P$ such that $P_{s\times t}$ equals a spatial average of all the $s\times t$ possible rectangles inside $K$ with a correction of the order $O\left(\frac{\max(s,t)}{\min(m,n)}\right)$. 

\begin{defin}[The symmetrization procedure]

Given $Q_K$, define a distribution $\hat{P}$ over the whole plane by tiling it with copies of the distribution $Q_{K}$. That is, for $K(i,j)\equiv K+(mi,n,j)$, $\hat{P}_{\cup_{ij} K(i,j)}(a_{\cup_{ij} K(i,j)})\equiv\prod_{i,j} Q_K(a_{K(i,j)})$. Then, the TI distribution $P$ given by

\be
P_{K'}\equiv\frac{1}{mn}\sum_{x=1}^n\sum_{y=1}^m\hat{P}_{K'+(x,y)}
\ee

\noindent will be called the \emph{symmetrization} of $Q$.

\end{defin}

That $P$ is TI invariant can be seen by noting that $\hat{P}$ is invariant under translations of $m$ ($n$) sites in the horizontal (vertical) axis. Random translations of $\{1,...,m\}$ and $\{1,...,n\}$ sites in those axes thus turn $\hat{P}$ into a 2D TI distribution.

It can be easily checked that 

\begin{align}
P_h=\frac{1}{(m-1)n}\sum_{x=0}^{m-2}\sum_{y=0}^{n-1} Q_{h+(x,y)}+O\left(\frac{1}{m}\right),\nonumber\\
P_v=\frac{1}{m(n-1)}\sum_{x=0}^{m-1}\sum_{y=0}^{n-2} Q_{v+(x,y)}+O\left(\frac{1}{n}\right).
\label{symmet_hv}
\end{align}

\noindent In general, 

\begin{align}
P_{s\times t}=&\frac{1}{(m-s)(n-t)}\sum_{x=0}^{m-s-1}\sum_{y=0}^{n-t-1} Q_{s\times t+(x,y)}+\nonumber\\
&+O\left(\frac{s}{m}\right)+O\left(\frac{t}{n}\right).
\label{symmet_gen}
\end{align}

Now we are ready to establish the equivalence between $\ENERGY$ and the computation of the minimum energy per site of a local Hamiltonian. Suppose that there exists a configuration $a_{\Z^2}$ of the square lattice such that $\lim_{N\to\infty}H_N(a_{N\times N})/N^2 =E$. Applying symmetrization over the distribution $Q_{N\times N}(b)=\delta_{b,a_{N\times N}}$, with energy-per-site $E_N$, one derives a 2D TI distribution $P$ with $\sum_{i=1}^n\sum_{a_{K_i}}F_i(a_{K_i})P_{K_i}(a_{K_i})=E_N+O(1/N)$. Taking the limit $N\to\infty$, we conclude that the energy value of any configuration can be matched by a TI marginal. Conversely, given any 2D TI distribution $P$ with $\sum_{i=1}^n\sum_{a_{K_i}}F_i(a_{K_i})P_{K_i}(a_{K_i})=E$, for any $N$ there exists, by convexity, an $N\times N$ square configuration $a_{N\times N}$ with $P_{N\times N}(a_{N\times N})\not=0$ such that $H_N(a_{N\times N})/N^2\leq E+O(1/N)$. It follows that the solutions of both problems are equal.

\section{Approximations of the set of TI marginals}
\label{properties}
In this section we will prove that, for any given scenario, the set of TI marginals admits an approximate characterization up to arbitrary accuracy.

The symmetrization protocol suggests a simple (but expensive) converging sequence of relaxations. Given $Q_K(a_K)$, with $K=s\times t$, a necessary condition for $Q_K(a_K)$ to be a TI marginal is that $Q_K(a_K)$ is the marginal of a distribution $P^{(n)}$ over the square $n\times n$, subject to the rules:

\begin{align}
&P^{(n)}_{(n-1)\times n+(x,0)}=P^{(n)}_{(n-1)\times n}, \mbox{ for } x=0,1,\nonumber\\
&P^{(n)}_{n\times (n-1)+(0,y)}=P^{(n)}_{n\times (n-1)}, \mbox{ for } y=0,1.
\label{mock_TI}
\end{align}

\noindent Intuitively, $P^{(n)}$ is modeling the marginal for the region $n\times n$ of an overall 2D TI distribution containing $Q_K$. 

The verification can be carried out via linear programming \cite{linear}. Linear programming is a branch of convex optimization concerned with the resolution of problems of the form

\begin{align}
&\max \bar{c}\cdot\bar{x},\nonumber\\
&\mbox{s.t. } A\bar{x}\geq \bar{b}, \bar{x}\geq 0.
\end{align}
\noindent where $\bar{c}\in\R^p$, the $q\times p$ matrix $A$ and $\bar{b}\in \R^q$ are the inputs of the problem; $\bar{x}\in \R^p$ are the problem variables; and $\bar{s}\geq 0$ is used to denote that all the components of the vector $\bar{s}$ are non-negative.

For each (primal) linear program there exists a dual problem

\begin{align}
&\min \bar{b}\cdot\bar{x},\nonumber\\
&\mbox{s.t. } A^T\bar{y}\geq \bar{c}, \bar{y}\geq 0.
\end{align}

\noindent Remarkably, the solutions of both primal and dual problems coincide. Numerical algorithms aimed at solving one problem hence run optimizations over the primal and dual problems. This allows the solver to give rigorous upper and lower bounds on the optimal solution. At present, there exist numerous free software implementations of interior-point methods for linear programs \cite{linear,Matousek2007}. In addition, all linear programs can be solved exactly via the costly Fourier-Motzkin elimination method \cite{fourier}. 

In our case, we need the solver to verify that there exists a probability distribution $P^{(n)}_{n\times n}$ for $n^2$ variables satisfying the linear conditions (\ref{mock_TI}), together with $P^{(n)}_{K}=Q_K$. This can be formulated as a linear program by regarding each probability $P^{(n)}_{n\times n}(a_K)$ as a free variable, and choosing $A,b$ so that $P^{(n)}_{n\times n}(a_K)$ satisfies the corresponding linear constraints. As for the objective function to optimize, we can take $\bar{c}=0$, i.e., the solution of the primal linear program will be zero provided that there exists a feasible point. If there exists a distribution $P^{(n)}_{n\times n}(a_K)$ compatible with $Q_K$, the solver will find it. Conversely, if there is no such distribution, then the solver will return a solution for the dual problem with an objective value smaller than $0$. Such a solution is, in effect, a \emph{certificate of infeasibility}, that is, a computer-generated proof that $Q_K$ is not a TI marginal.

The method described above is a relaxation of the property of being a TI marginal: if $Q_K$ does not pass the $n^{th}$ test, then it clearly is not a TI marginal. If, on the other hand, $Q_K$ passes the $n^{th}$ test, then we can apply the symmetrization protocol over the distribution $P_{n\times n}$ and obtain a 2D TI distribution $\hat{P}$ that, due to eq. (\ref{symmet_gen}) and the condition $P_K=Q_K$, satisfies $\hat{P}_K=Q_K+O\left(\frac{\max(s,t)}{n}\right)$. 

This algorithm is highly inefficient, though, since the time and space complexity of the computations scales as $O(e^{\alpha n^2})$. Actually, there is a much more practical relaxation of the set of TI marginals achieving the same accuracy that just involves $O(e^{\beta n})$ operations. 

\begin{defin}[Approximate solution to the marginal problem]
\label{approx_sol}

Given $Q_K$, with $K=s\times t$, $t\leq s$, verify that $Q_K$ is the marginal of $P^{(n)}_{n\times t}$, with 

\begin{align}
&P^{(n)}_{(n-1)\times t+(1,0)}=P^{(n)}_{(n-1)\times t},\nonumber\\
&P^{(n)}_{n\times (t-1)+(0,1)}=P^{(n)}_{n\times (t-1)}.
\label{mock_TI2}
\end{align}
\end{defin}
\noindent This can again be formulated as a linear program, and is obviously a relaxation of the property of being a TI marginal. The condition (\ref{mock_TI2}) will appear often in the rest of the article, so we will give it a name. Any distribution $P_{n\times t}$ in the rectangle $n\times t$ satisfying the above conditions will be called \emph{locally translation invariant} (LTI).

To see that the relaxation above is $O(s/n)$-close to the actual set of TI marginals, suppose that $Q_K$ is the marginal of $P^{(n)}_{n\times t}$. We will next extend $P^{(n)}_{n\times t}$ to a distribution $P'_{n\times n}$ with the property that the marginal of any $s\times t$ rectangle equals $Q_K$. 

In order to derive $P'_{n\times n}$, we regard $P^{(n)}_{n\times t}$ as the $t$-site marginal of a TI 1D system with local variables of dimension $d^n$. We can do so because $P^{(n)}_{n\times t}$ satisfies the second line of (\ref{mock_TI2}). From the triviality of the marginal problem for TI 1D systems, we know that there must exist a distribution $P'_{n\times n}$ satisfying the afore-mentioned properties. Finally, it is easy to see, from eq. (\ref{symmet_hv}), that the symmetrization of $P'_{n\times n}$ will be a 2D TI distribution $\hat{P}$ with marginal $\hat{P}_K=Q_K+O(s/n)$.

To conclude, we would like to remark that the above sequence of relaxations of the set of marginals also allows the user to approximately solve $\ENERGY$. Indeed, given $\{(F_i,K_i)\}_i$, one can carry out, via linear programming, the optimization:

\begin{align}
E^n\equiv &\min\sum_i\sum_{a_{K_i}}F_i(a_{K_i})P^{(n)}_{K_i}(a_{K_i}),\nonumber\\
\mbox{s.t. }&P^{(n)}_{(n-1)\times t+(x,0)}=P^{(n)}_{(n-1)\times t}, \mbox{ for } x=0,1,\nonumber\\
&P^{(n)}_{n\times (t-1)+(0,y)}=P^{(n)}_{n\times (t-1)}, \mbox{ for } y=0,1.
\end{align}

\noindent From what we reasoned above, it follows that $E^n\leq \ENERGY(\{(F_i,K_i)\}_i)\leq E^n+ O(1/n)$.

\section{The exact marginal problem in 2D: characterizations}
\label{positive}
In this section we will identify certain variants or scenarios where the TI marginal problem is exactly solvable. That is, where there exists an algorithm that will determine with certainty if the given distributions are TI marginals or not.

\subsection{The exact marginal problem with reflection symmetry}

Let us start with a relevant variant of the marginal problem: the 2D TI marginal problem with reflection symmetry. This is the case where, in addition to demanding the existence of a TI extension, we require this to be invariant under reflections on both axes. Since Nature is approximately invariant under parity reflection, this condition holds for the thermal states of many physically relevant Hamiltonians, such as the isotropic Ising model and the Potts model \cite{ising, potts}.

\begin{defin}[The $s\times 2$ TI marginal problem with reflection symmetry]

Let $K=s\times 2$. Given $Q_K(a_K)$, determine if there exists a 2D TI distribution $P$ such that $P_{K}(a_{K})=Q(a_{K})$ and $P_{n\times n}(a_{n\times n})=P_{n\times n}(a^H_{n\times n})=P_{n\times n}(a^V_{n\times n})$, for all $n$. Here, $a^H_{n\times n}$, $a^V_{n\times n}$ denote, respectively, the permutation of the variables $(a_{\{(x,y)\}})|_{0\leq x\leq s-1,0\leq y\leq 1}$ associated to a reflection over the horizontal and vertical axis, respectively.
\end{defin}

\noindent Interestingly, the $s\times 2$ TI marginal problem with reflection symmetry can be solved exactly. 

\begin{prop}
\label{symProp}
The existence of a 2D TI distribution with reflection symmetry is equivalent to the conditions:

\begin{align}
&Q_{K}(a_K)=Q_{K}(a^H_K)=Q_{K}(a^V_K),\nonumber\\
&Q_{(s-1)\times 2}=Q_{(s-1)\times 2+(1,0)}.
\label{symmetTI}
\end{align}

\end{prop}

\begin{proof}
The second condition implies that, viewed as a 1D system of $s$ sites with local variables of dimension $d^2$, $Q_{s\times 2}$ is the marginal of a TI system. In particular, for any $n$ one can find $P^1_{n\times 2}$ with the property that $P^1_{K+(x,0)}=Q_{K}$ for $0\leq x\leq n-s$. This property is kept if we make the distribution invariant under reflection on both axes.

\begin{align}
P^2_{2\times n}(a_{2\times n})=&\frac{1}{4}(\hat{P}_{n\times 2}(a_{n\times 2})+\hat{P}_{n\times 2}(a^H_{n\times 2})+\nonumber\\
&\hat{P}_{v}(a^V_{n\times 2})+\hat{P}_{n\times 2}(a^{HV}_{n\times 2})).
\label{reflection_sym}
\end{align}

Reflection under the horizontal axis implies, in particular, that $P^2_{n\times 1}=P^2_{n\times 1+(0,1)}$. Therefore, we can view $P^2$ as the $2$-site marginal of a 1D TI system with variables of local dimension $d^{n}$ and extend it to an $n\times n$ square. From that point on, we can consider the symmetrization of this distribution $P^3_{n\times n}$, which will return a 2D TI reflection-symmetric distribution with $Q_K+O(1/n)$ as a marginal. Invoking the closure of the set of marginals of 2D TI distributions \cite{Goldstein2017}, we conclude that $Q_K$ admits a 2D TI extension $P$. Finally, since $Q_K$ satisfies reflection symmetry, one can choose $P$ to be symmetric as well.

\end{proof}

The characterization (\ref{symmetTI}) of the set of $s\times 2$ TI marginals with reflection symmetry implies that, for $F_K$ satisfying $F_K(a_K)=F_K(a^H_K)=F_K(a^V_K)$, $\ENERGY(F_K)$ is exactly solvable. Indeed, let $\hat{P}_K$ be any TI marginal and call $E$ the value of the functional in eq. (\ref{energy_functional}) evaluated in $\hat{P}_K$. Then it is easy to see that $P^2_{s\times 2}$, as defined by eq. (\ref{reflection_sym}) (replacing $n$ by $s$), admits a TI, reflection-symmetric extension. Moreover, the value of the functional is also $E$. It follows that, for $F_K(a_K)=F_K(a^H_K)=F_K(a^V_K)$, one can assume that the minimizer of $\ENERGY(F_K)$ is a TI, reflection-symmetric marginal. Hence one can use linear programming to solve the problem.

\subsection{Binary local random variables}

Take the local dimension of the random variable at each site to be $d=2$ (bits), and suppose that we just want to characterize the bivariate distributions $P_h(a,b), P_v(a,b)$ between nearest-neighbors. The next proposition states that the problem is solvable even for lattices of spatial dimensions higher than $2$. 

\begin{prop}
Let $C=\{(1,0,0,...),(0,1,0,...),...\}$ be the set of all $k$-dimensional vectors with null components except one entry with value $1$, and let $\{P_{\{\bar{0},\bar{c}\}}(a,b): \bar{c}\in C\}$ be all nearest-neighbor marginals of a hypercubic lattice of spatial dimension $k$, with $a,b\in\{0,1\}$. Then \textbf{MARGINAL}$(\{P_{\{\bar{0},\bar{c}\}}(a,b): \bar{c}\in C\})$ can be formulated as a linear program. Moreover, for spatial dimensions $k=2,3$, the existence of a TI extension for $\{P_{\{\bar{0},\bar{c}\}}(a,b): \bar{c}\in C\}$ is equivalent to LTI, i.e., the condition that

\be
\sum_{b} P_{\{\bar{0},\bar{c}\}}(x,b)=\sum_{a} P_{\{\bar{0},\bar{c}'\}}(a,x),
\label{trivial2}
\ee
\noindent for all $\bar{c},\bar{c}'\in C$.
\end{prop}

For $d=2$ and $k=2,3$ the nearest-neighbor marginal problems are thus trivial. This perhaps explains why the only known solvable classical models in 2D are bit models.

\begin{proof}

For simplicity, consider the $2D$ case $k=2$. Then, for $d=2$, LTI on either $P_h$ or $P_v$ (i.e., $\sum_{a}P_h(a,x)=\sum_{a}P_h(x,a)$, $\sum_{a}P_v(a,x)=\sum_{a}P_v(x,a)$) implies that both distributions $P_h$ and $P_v$ are symmetric. Now, suppose that the pair $(P_h,P_v)$ admits a 2D TI extension. If we make this extension reflection-invariant -via eq. (\ref{reflection_sym})-, the distributions $P_h,P_v$ will not change. In other words, $(P_h, P_v)$ is a TI marginal iff it is a TI reflection-symmetric marginal. From Proposition \ref{symProp}, we know that the marginals of 2D TI reflection-symmetric distributions correspond to the marginals of symmetric, LTI squares. This allows us to completely characterize the set of TI marginals $(P_h, P_v)$ via linear programming. The above symmetrization argument holds not only in 2D, but in any spatial dimension, and so does the proof of Proposition \ref{symProp}. 

It follows that the set of nearest-neighbor marginals in any dimension is described by a convex polytope, i.e., a set defined by a finite number of linear inequalities or \emph{facets} $\{F^i_h(P_h)+F^i_v(P_v)+...\leq 0\}_{i=1}^n$, where $F^i_h, F^i_v,...,$ are linear functionals on the probabilities $P_h(a,b)$, $P_v(a,b)$, etc. Using standard combinatorial software \cite{panda}, we managed to derive the facets $\{F^i_h(P_h)+F^i_v(P_v)\leq 0\}_{i=1}^n$ which define the 2D set. We verified, using linear programming, that the set all distributions $(P_h, P_v)$ with 

\begin{align}
&\sum_{a}P_h(a,x)=\sum_{a}P_h(x,a)=\nonumber\\
&\sum_{a}P_v(a,x)=\sum_{a}P_v(x,a)
\label{trivial}
\end{align}

\noindent cannot violate any of them. This implies that the above conditions, namely, LTI, characterize completely the set of nearest-neighbor TI marginals. Similarly, we verified that the analog 3D problem, with input $\{P_{\{\vec{0},\hat{c}\}}(a,b): \hat{c}=(1,0,0), (0,1,0),(0,0,1)\}$, reduces to verifying that conditions (\ref{trivial2}) hold.
\end{proof}

We now move to the problem of characterizing the distributions $(P_h,P_v,P_{+},P_{-})$ corresponding to horizontal and vertical nearest-neighbors and north-east $(+)$ and south-east $(-)$ next-to-nearest neighbor distributions in $2D$. This problem is not trivial (i.e., it does not reduce to verifying LTI), since the distribution $P_h(a,b)=P_v(a,b)=P_+(a,b)=P_-(a,b)=\frac{1}{2}\delta_{a\oplus b,1}$ satisfies $\sum_aP_s(a,x)=\sum_aP_t(x,a)$ for $s,t=h,v,+,-$ and yet it does not admit a TI extension~\footnote{Consider the random variables $a_{0,0},a_{1,0},a_{0,-1}$ and imagine that there exists a probability distribution $P$ for the three of them, with $P_h, P_v, P_{+}$ as marginals. Then, $\frac{1}{2}=P_v(a_{0,0}=0,a_{0,-1}=1)=P(a_{0,0}=0,a_{1,0}=0,a_{0,-1}=1)+P(a_{0,0}=0,a_{1,0}=1,a_{0,-1}=1)$. However, by $P_v(0,0)=0$ ($P_{+}(1,1)=0$), the first (second) term of the right hand side must equal zero. We thus reach a contradiction.}. And yet, as the next result shows, the set of nearest and next-to-nearest neighbor marginals can also be characterized via linear programming.

\begin{prop}
\label{nextToProp}
For $d=2$, the nearest and next-to-nearest neighbors marginals $P_h,P_v,P_{+},P_{-}$ admit a TI extension iff they are marginals of a distribution $P_{2\times 2}$ satisfying LTI.
\end{prop}

\begin{proof}

If $P_h,P_v,P_{+},P_{-}$ are TI marginals, then they must constitute an approximate solution of the marginal problem, in the sense explained in Definition \ref{approx_sol}. In particular, they must belong to the polytope of distributions $P_h,P_v,P_{+},P_{-}$ admitting a $2\times 2$ LTI extension. We find, using the combinatorial software Panda \cite{panda}, that this polytope has $13$ extreme points, each of which belongs to $6$ classes modulo rotations, reflections and relabelings of the variables:

\begin{enumerate}[label=(C{\arabic*})]
\item 
The two points in this class are deterministic: $P_h(a,b)=P_v(a,b)=P_{+}(a,b)=P_{-}(a,b)=\delta_{a,0}\delta_{b,0}$ and $P_h(a,b)=P_v(a,b)=P_{+}(a,b)=P_{-}(a,b)=\delta_{a,1}\delta_{b,1}$. They obviously admit a deterministic 2D TI extension.

\item 
This class has two elements of the form
\begin{align}
&P_{h}(a,b)=\frac{1}{2}\delta_{a,b\oplus s},P_{v}(a,b)=\frac{1}{2}\delta_{a\oplus 1,b\oplus s}\\
&P_{+}(a,b)=P_{-}(a,b)=\frac{1}{2}\delta_{a\oplus 1,b},
\end{align}
\noindent with $s=0,1$.

The two elements are generated by applying the symmetrization process to the $2\times 2$ deterministic distributions

\be
\begin{array}{|c|c|}
\hline
0&0\\
\hline
1&1\\
\hline
\end{array},
\begin{array}{|c|c|}
\hline
1&0\\
\hline
1&0\\
\hline
\end{array}.
\ee

\item
This class contains just one element, namely:
\begin{align}
&P_{h}(a,b)=P_{v}(a,b)=\frac{1}{2}\delta_{a\oplus 1,b}\\
&P_{+}(a,b)=P_{-}(a,b)=\frac{1}{2}\delta_{a,b}.
\end{align}

It is generated by the $2\times 2$ deterministic distribution

\be
\begin{array}{|c|c|}
\hline
1&0\\
\hline
0&1\\
\hline
\end{array}.
\ee

\item
The class has two elements of the form

\begin{align}
&P_h(a,b)=P_v(a,b)=P_{+}(a,b)=P_{-}(a,b)=\frac{1}{4}(\delta_{a,s}+\delta_{b,s}),
\end{align}
\noindent with $s=0,1$. They are respectively generated by symmetrizing the $2\times 2$ deterministic distributions:
\be
\begin{array}{|c|c|}
\hline
1&0\\
\hline
0&0\\
\hline
\end{array},\begin{array}{|c|c|}
\hline
0&1\\
\hline
1&1\\
\hline
\end{array}.
\ee

\item
This class has four elements, generated via applying reflections to the generator:
\begin{align}
&P_h(a,b)=P_v(a,b)=P_{+}(a,b)=\frac{1}{3}\delta_{a\cdot b,0},\\
&P_{-}(a,b)=\frac{2}{3}\delta_{a,0}\delta_{b,0}+\frac{1}{3}\delta_{a,1}\delta_{b,1}.
\end{align}

The latter, in turn, can be generated by the $3\times 3$ deterministic distribution
\be
\begin{array}{|c|c|c|}
\hline
0&0&1\\
\hline
0&1&0\\
\hline
1&0&0\\
\hline
\end{array}.
\ee

\item
This class has two elements, generated either by applying the identity or a $\pi/2$ rotation over the generator:
\begin{align}
&P_h(a,b)=P_v(a,b)=\frac{1}{4},\\
&P_{+}(a,b)=\frac{1}{2}\delta_{a\oplus b,1},P_{-}(a,b)=\frac{1}{2}\delta_{a\oplus b,0}.
\end{align}

It can be verified that the generator in turn can be built from the $4\times 4$ distribution:

\be
\begin{array}{|c|c|c|c|}
\hline
1&1&0&0\\
\hline
1&0&0&1\\
\hline
0&0&1&1\\
\hline
0&1&1&0\\
\hline
\end{array}.
\ee

\end{enumerate}

\noindent Since all these points are TI marginals, any convex combination thereof will also be a TI marginal. It follows that, for $d=2$, the first level of the hierarchy presented in Def. \ref{approx_sol} already characterizes all nearest and next-to-nearest neighbor TI marginals. 

\end{proof}

From Proposition \ref{nextToProp} one can infer the exact solvability of the $\ENERGY$ problem for TI Hamiltonians with nearest and next-to-nearest neighbor interactions. Indeed, since any TI marginal is a convex combination of the above $13$ points, one just needs to evaluate the energy per site of each of them and take the smallest result. It is worth noting that, for each extreme point $P$ listed above, there exists a unique (modulo translations) deterministic distribution $\tilde{P}$ for the whole lattice such that $\mbox{supp}(\tilde{P}_{g+z})\subset \mbox{supp} P_{g}$, with $g=h,v,+,-$ and $z\in\Z^2$. Now, let the value of $\ENERGY(F_h,F_{v},F_{+},F_{-})$ be achieved by just one extreme point $P$ of the set of TI marginals. If $P$ is generated by tiling the plane with an irreducible rectangle of size $m\times n$, then the number of lattice configurations achieving the minimum energy per site of the corresponding Hamiltonian is $mn$. Assuming that, at zero temperature, the system is in an equal mixture of all such configurations, its entropy will be $S_0=\ln(mn)$. As we will see in the next section, the uniqueness of the extended distribution on the whole lattice is lost when each site can take more than two values. This is in contrast to the situation in one dimension, where maximal entropy extensions of LTI distributions exist and are unique~\cite{Goldstein2017}.

\subsection{Ternary local random variables}
The purpose of this section is to characterize the set of TI nearest-neighbor marginals $P_h,P_v$ for ternary local random variables ($d=3$). For $d=2$, all such $P_h$ and $P_v$ can be characterized simply by imposing (\ref{trivial2}). For $d=3$, this condition is not sufficient. Take, for example, the distribution $P_h(0,2)=P_h(1,0)=P_h(2,1)=P_v(0,2)=P_v(1,2)=P_v(2,1)=\frac{1}{3}$. It can be easily checked that this distribution satisfies eq. (\ref{trivial2}). However, one can verify numerically that it does not belong to the first level of the hierarchy of approximations proposed in Definition \ref{approx_sol}. 

Luckily, this first approximation turns out to be enough to characterize the nearest-neighbor TI marginals. 

\begin{prop}
\label{nextToProp2}
For $d=3$, the nearest-neighbor marginals $P_h,P_v$ admit a TI extension iff they are compatible with a distribution $P_{2\times 2}$ satisfying LTI.
\end{prop}

\begin{proof}

The proof follows the same lines as the proof of Proposition \ref{nextToProp}. To find the extreme points/facets of the polytope of marginals admitting a $P_{2\times 2}$ LTI extension, we first use a linear program to maximize random objective functions under such a set. The vectors thus obtained are either vertices of the polytope or lie on the facets. By computing the dual description of this set of vectors, we obtain a set of linear inequalities some of which may be facets of the polytope we seek. To check whether the inequalities are indeed facets we use linear programming again but this time taking the vectors which define the inequalities (without the bounds) as objective functions. If an inequality is indeed a facet then maximizing this objective function using the same constrains as before will give us the upper bound of the inequality as the value of the objective function. If it is not a facet then this maximization will violate the upper bound of the inequality, at which time we can add the vector achieving this maximization to the list of potential vertices. By iterating this violation/addition procedure until no inequality can be violated when performing the second maximization, we obtain all the facets of the polytope, from which the extreme points can also be computed. In the $d=3$ case, the polytope is defined by 98 extreme points, which fall into 10 equivalence classes after taking local permutations of outcomes and reflections into account. A representative from each class is given below, with $P_h(a,b)$ and $P_v(a,b)$ written as a vector whose $i$th coordinate in base $3$ gives the values of $a,b$: $P_{h,v}(a,b)\equiv(0_a0_b,0_a1_b,0_a2_b,\ldots,2_a2_b)$. It can be readily checked that each of these class representatives can be generated by symmetrizing the depicted deterministic distributions accompanying it.

\begin{enumerate}[label=(C{\arabic*})]

\item $\begin{array}{|c|}
\hline
2\\
\hline
\end{array}$

\begin{align*}
P_h&=(0,0,0,0,0,0,0,0,1),\\
P_v&=(0,0,0,0,0,0,0,0,1).
\end{align*}

\item $\begin{array}{|c|c|}
\hline
1&2\\
\hline
2&1\\
\hline
\end{array}$
\begin{align*}
P_h&=(0,0,0,0,0,\frac{1}{2},0,\frac{1}{2},0),\\
P_v&=(0,0,0,0,0,\frac{1}{2},0,\frac{1}{2},0).
\end{align*}

\item $\begin{array}{|c|c|}
\hline
1&2\\
\hline
\end{array}$
\begin{align*}
P_h&=(0,0,0,0,0,\frac{1}{2},0,\frac{1}{2},0),\\
P_v&=(0,0,0,0,\frac{1}{2},0,0,0,\frac{1}{2}).
\end{align*}

\item $\begin{array}{|c|c|}
\hline
0&2\\
\hline
2&2\\
\hline
1&1\\
\hline
\end{array}$
\begin{align*}
P_h&=(0,0,\frac{1}{6},0,\frac{1}{3},0,\frac{1}{6},0,\frac{1}{3}),\\
P_v&=(0,0,\frac{1}{6},\frac{1}{6},0,\frac{1}{6},0,\frac{1}{3},\frac{1}{6}).
\end{align*}

\item $\begin{array}{|c|c|}
\hline
0&2\\
\hline
2&1\\
\hline
1&1\\
\hline
\end{array}$

\begin{align*}
P_h&=(0,0,\frac{1}{6},0,\frac{1}{3},\frac{1}{6},\frac{1}{6},\frac{1}{6},0),\\ P_v&=(0,0,\frac{1}{6},\frac{1}{6},\frac{1}{6},\frac{1}{6},0,\frac{1}{3},0).
\end{align*}

%

\item $\begin{array}{|c|c|}
\hline
0&2\\
\hline
2&1\\
\hline
2&0\\
\hline
1&2\\
\hline
\end{array}$
\begin{align*}
P_h&=(0,0,\frac{1}{4},0,0,\frac{1}{4},\frac{1}{4},\frac{1}{4},0),\\ P_v&=(0,0,\frac{1}{4},\frac{1}{4},0,0,0,\frac{1}{4},\frac{1}{4}).
\end{align*}

\item $\begin{array}{|c|c|}
\hline
0&2\\
\hline
1&2\\
\hline
\end{array}$
\begin{align*}
P_h&=(0,0,\frac{1}{4},0,0,\frac{1}{4},\frac{1}{4},\frac{1}{4},0),\\
P_v&=(0,\frac{1}{4},0,\frac{1}{4},0,0,0,0,\frac{1}{2}).
\end{align*}

\item $\begin{array}{|c|c|c|c|}
\hline
2&1&0&2\\
\hline
2&1&1&0\\
\hline
\end{array}$
\begin{align*}
P_h&=(0,0,\frac{1}{4},\frac{1}{4},\frac{1}{8},0,0,\frac{1}{4},\frac{1}{8}),\\ P_v&=(0,\frac{1}{8},\frac{1}{8},\frac{1}{8},\frac{1}{4},0,\frac{1}{8},0,\frac{1}{4}).
\end{align*}

\item $\begin{array}{|c|c|c|}
\hline
0&2&1\\
\hline
2&1&0\\
\hline
1&0&2\\
\hline
\end{array}$
\begin{align*}
P_h&=(0,0,\frac{1}{3},\frac{1}{3},0,0,0,\frac{1}{3},0),\\
P_v&=(0,0,\frac{1}{3},\frac{1}{3},0,0,0,\frac{1}{3},0).
\end{align*}

\item $\begin{array}{|c|c|c|}
\hline
0&2&1\\
\hline
\end{array}$
\begin{align*}
P_h&=(0,0,\frac{1}{3},\frac{1}{3},0,0,0,\frac{1}{3},0),\\
P_v&=(\frac{1}{3},0,0,0,\frac{1}{3},0,0,0,\frac{1}{3}).
\end{align*}

\end{enumerate}

\end{proof}

\section{The exact marginal problem in 2D: no-go theorems}
\label{negative}
In view of the results of the previous section, one would imagine that the exact resolution of the marginal problem for scenarios with high local dimension $d$ is merely a matter of computational power. Note that all the marginal sets characterized so far satisfy the following properties:

\begin{enumerate}
\item
They allow us to solve $\ENERGY$ exactly.
\item
When we represent them in probability space, they happen to be convex polytopes, i.e., sets determined by a finite number of linear inequalities. More broadly, they are \emph{semi-algebraic sets}. A subset ${\mathcal{S}}$ of $\R^s$ is a \emph{basic closed semi-algebraic set} if there exist a finite number of polynomials $\{F_i\}_{i=1}^u$ on $\bar{x}\in\R^s$ and the slack vector $\bar{y}\in\R^t$ such that

\begin{align}
\bar{x}\in{\cal S} \mbox{ iff }&\exists\bar{y}\in\R^{t}, s.t.\nonumber\\
&F_i(\bar{x}, \bar{y})\geq 0, i=1,...,u.
\label{basic_closed}
\end{align}
\noindent All subsets of $\R^t$ which one can characterize via linear programming \cite{linear} or the more general tool of semidefinite programming \cite{sdp} fall within this category. A semi-algebraic set would be the union of a finite number of basic closed semi-algebraic sets. However, it is easy to show that any closed convex semi-algebraic set is also basic, so in the following we will use both terms interchangeably.

\end{enumerate}

\noindent Extrapolating, one would expect that the marginal sets of 2D TI distributions with arbitrary $d$ should retain at least one of the above features.

In the following pages we show that this is not the case even when we aim at solving the simplest non-trivial marginal problem: the characterization of $(P_h, P_v)$.

\subsection{Undecidability of $\ENERGY$}
In this section we will prove the following result:

\begin{theo}
There exists no algorithm to solve $\ENERGY(F_h,F_v)$ for arbitrary local dimension $d$ and $F_h,F_v:\{0,...,d-1\}^2\to \{0,1\}$.
\end{theo}

The proof of this theorem, as well as the proof of the next no-go result, will rely heavily on certain mathematical results on tilings of the plane, so let us first introduce a basic vocabulary.

Let $A$ be a finite alphabet. Any function $f:\Z^2\to A$ defines a \emph{tiling of the plane} with the set of tiles $A$. Any pair of subsets $\T=(\T_h, \T_v)$ of $A\times A$ defines a \emph{tiling rule}. We say that a tiling $f$ respects the tiling rule $\T$ if, for all $(x,y)\in\Z^2$,

\be
(f(x,y),f(x+1,y))\in \T_h, (f(x,y),f(x,y+1))\in \T_v.
\ee
\noindent If the rule $\T$ is implicitly known, we call $f$ a valid tiling.

Certain rules $\T$ do not admit any valid tiling, e.g.: take $A=\{0,1\}$ and $\T_h=\{(0,1)\},\T_v=A\times A$. The problem of deciding if a given rule $\T$ admits a valid tiling was first raised by Wang \cite{wang}, and proven undecidable by Berger \cite{berger}, who also showed the existence of tiling rules $\T$ admitting just aperiodic tilings. The construction used by Berger to derive the latter result uses 20426 tiles. This number has been decreasing over the years as new aperiodic tiling rules requiring less tiles were found. The current record, held by Jeandel and Rao, uses only 11 tiles~\cite{jeandel}.

We will prove that $\ENERGY(F_h,F_v)$ is undecidable by reducing it to the general tiling problem. Let $\T$ be a tiling rule, and define the input of $\ENERGY$ as

\begin{align}
F_h(a,b)=&1,\mbox{ if } (a,b)\in\T_h,\nonumber\\
&0,\mbox{ otherwise},\nonumber\\
F_v(a,b)=&1,\mbox{ if } (a,b)\in\T_v,\nonumber\\
&0,\mbox{ otherwise}.
\end{align}

\noindent We claim that $\ENERGY(F_h,F_v)=2$ iff $\T$ admits a valid tiling. Consequently, the $\ENERGY$ problem is undecidable.

First, suppose that $\T$ admits a valid tiling. We choose an $n\times n$ square of this tiling, and apply symmetrization to obtain the marginals $P^{(n)}_h,P^{(n)}_v$ of a 2D TI distribution. Those will satisfy $F_h(P^{(n)}_h)+F_v(P^{(n)}_v)=2-O(1/n)$. That way, we obtain a sequence of TI marginals $(P^{(n)}_h,P^{(n)}_v)_n$ whose energy per site tends to $2$. The closure of the set of TI marginals implies that there exists a TI marginal $(P^\star_h, P^\star_v)$ with $F_h(P^\star_h)+F_v(P^\star_v)=2$. Since $F_h(P),F_v(P)\leq 1$ for any distribution, this implies that $\ENERGY(F_h,F_v)=2$.

Conversely, suppose that $\ENERGY(F_h,F_v)=2$. Then, there exists a 2D TI distribution $P$ saturating the bound. Given $P_{|n\times n|}$, take any configuration $a_{|n\times n|}$ such that $P_{|n\times n|}\not=0$. Obviously, $a_{|n\times n|}$ is a valid tiling for $|n\times n|$. Now, consider $P_{|(n+1)\times (n+1)|}$ and take any configuration $a_{|(n+1)\times (n+1)|}$ such that $a_{|n\times n|}$ is an inner square of $a_{|(n+1)\times (n+1)|}$ and $P_{|(n+1)\times (n+1)|}(a_{|(n+1)\times (n+1)|})\not=0$. That such a configuration must exist follows from the fact that $\sum_{a'}P_{|(n+1)\times (n+1)|}(a_{|n\times n|},a')=P_{|n\times n|}(a_{|n\times n|})\not=0$, where $a'$ denotes the local variables of the inner border of the $|(n+1)\times (n+1)|$ square. Iterating, we obtain a sequence $(a_{|n\times n|})_n$ of valid tiles for ever-growing squares with the particularity that, for all $n$, $a_{|n\times n|}$ is contained in $a_{|(n+1)\times (n+1)|}$. The desired tiling of the plane is hence given by the function $f$ assigning to each point $(x,y)$, with $|x|<n,|y|<n$ the corresponding tile in $a_{|n\times n|}$.

\subsection{The set of TI nearest-neighbors marginals is not a semi-algebraic set}

The goal of this section is to prove the following theorem.

\begin{theo}
\label{egregius}
For $d=2947$, the set of nearest-neighbor TI marginals $P_h,P_v$ is not semi-algebraic. Moreover, some pieces of its boundary are smoothly curved.
\end{theo}

Before proving this theorem, we want to remark that the value of $d=2947$ comes from the number of tiles required by our construction below. It is not necessarily tight: for all we know, the critical value $d$ which marks the transition from polytopes to curved sets could be as low as $4$. It is possible that, similar to the history of the minimum set of Wang tiles discussed above, our upper bound $d=2947$ be improved in the future. A consequence of Theorem \ref{egregius} is that, for $d\geq 2947$, the set of nearest-neighbor TI marginals cannot be characterized by either linear or semidefinite programming \cite{sdp}.

\begin{proof}
The idea of the proof is as follows: given a finite alphabet $A$, we will define a tiling rule $\T$, and we will consider the set $\P$ of all TI marginals $(P_h(a,b),P_v(a,b))$, with $a,b\in A$ such that

\be
\sum_{(a,b)\in\T_h}P_h(a,b)=\sum_{(a,b)\in\T_v}P_v(a,b)=1.
\ee

That is, we will consider the intersection between the set of TI marginals and two planes defined by integer coefficients. If the set of all TI marginals $(P_h,P_v)$ were semi-algebraic or a polytope, then so would $\P$.

We will then define two linear 1-site parameters

\be
\omega\equiv\sum_{a\in A}\hat{\omega}(a)P_{(0,0)}(a), \eta\equiv\sum_{a\in A}\hat{\eta}(a)P_{(0,0)}(a),
\ee

\noindent and characterize the set ${\cal S}$ of feasible values $(\omega,\lambda)$ in $\P$. We will find that the boundary of ${\cal S}$ contains both flat and a smoothly curved pieces, so we will conclude that the set of nearest-neighbor TI marginals with variables of dimension $d=|A|$ does not form a polytope. Similarly, by proving that ${\cal S}$ is not a basic closed semi-algebraic set, we will demonstrate that neither is the set of TI marginals.

The proof relies heavily on the connection between aperiodic tilings and immortal points of dynamical systems first pointed out in~\cite{kari1} and later extended in~\cite{kari2}. Before giving our implementation which proves the statement above we will briefly review the general construction given in~\cite{kari1,kari2}.

Let $M$ be a $2\times 2$ matrix with rational coefficients; and $c$, a rational vector in $\R^2$. We will consider a number $u$ of unit squares $\{R^i\}_{i=1}^u$ in $\R^2$, each of them defined by their integer corners $U^i=\{(m^i, n^i),(m^i+1, n^i),(m^i, n^i+1),(m^i+1, n^i+1)\}$. Let $R\equiv \cup_{i=1}^uR^i$. The tiles used in this construction are all derived from the prototile given in Fig.~\ref{prototile}. A tile is labeled by four vectors in $\R^2$, called $t, b, l, r$ (\emph{top, bottom, left and right}) and a \emph{region} $i\in \{1,...,u\}$. The components of $l$ and $r$ are rational, $t\in U^i$ and $b\in U^j$ for some $j=1,...,u$. A valid tile also satisfies the relation $f(t)+l=b+r$, where $f:\R^2\to\R^2$ is the affine function $f(\bar{z})\equiv M\bar{z}+c$.

Given any set $A$ (finite or infinite) of tiles so labeled and satisfying such constraints, the tiling rules we will consider are:
\begin{enumerate}
\item
To the right of any tile with region $i$ and right vector $r$, there can only be a tile with region $i$ and left vector $r$.
\item
To the bottom of any tile with bottom vector $b$, there can only be a tile with top vector $b$.
\end{enumerate}

\begin{figure}
  \centering
  \includegraphics[width=0.2\textwidth]{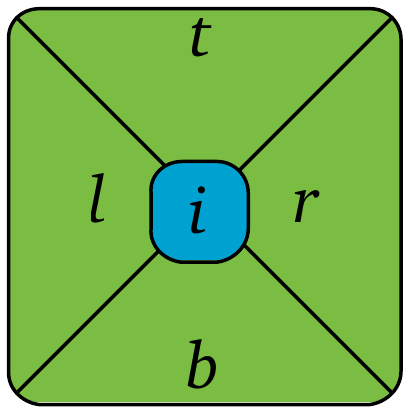}
  \caption{\textbf{The Kari prototile.}}
  \label{prototile}
\end{figure}

The first rule can be seen to imply that, for any valid row of tiles $k=1,...,m$ of the form $(i, t_k,b_k,l_k,r_k)$,

\be
f(\langle t\rangle)+\frac{l_1}{m}=\langle b\rangle+\frac{r_m}{m},
\ee
\noindent where $\langle t\rangle,\langle b\rangle$, are, respectively, the arithmetic means of the top and bottom vectors of the $n$ tiles. Note that, by convexity, $\langle t\rangle\in R^i$, $\langle b\rangle\in \mbox{conv}(R)$.

\begin{figure}
  \centering
  \includegraphics[width=0.6\textwidth]{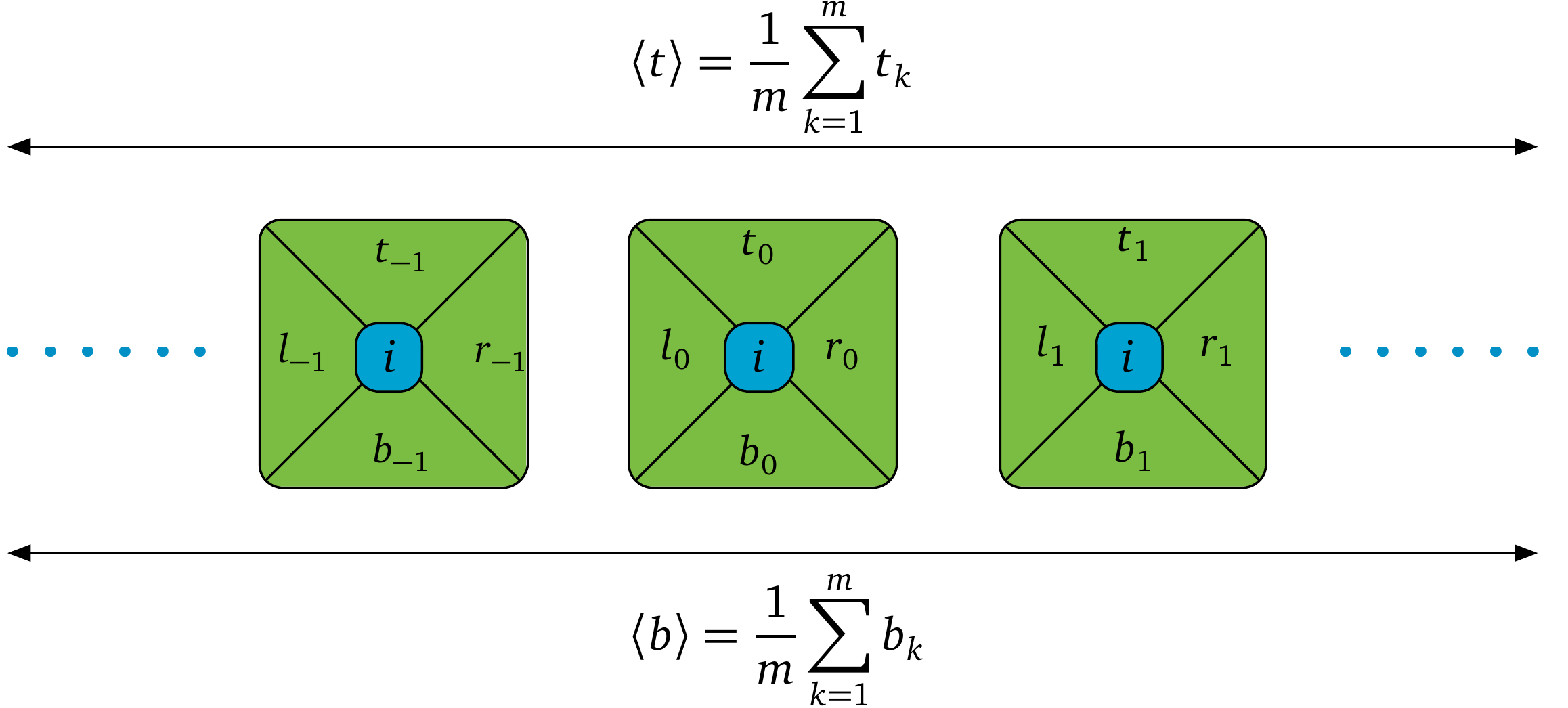}
  \caption{\textbf{Law of averages for a segment of Kari tiles.} For large $m$, $f(\langle t\rangle)\approx\langle b\rangle$.}
  \label{row}
\end{figure}

Let $A$ be such that the set $L$ of possible left vectors $l$, viewed as a subset of $\R^2$, is bounded and equal to the set of feasible right vectors. It follows that, for $n$ sufficiently large, we can view each valid row as a vector $\bar{z}\in R^i$ undergoing the transformation $\bar{z}\to M\bar{z}+c$. Moreover, by the second tiling rule, given the validly tiled rectangle $|m\times n|$, with $m\gg n\gg 1$, the sequence of top vector averages on each row $j$ can be interpreted as an orbit inside $R$ given by $(\bar{z}_j=f^{-j}(\bar{z}_{0})+O(\frac{n}{m}):j=-n,...,n)$. In addition, the region of the $j^{th}$ row indicates which square $R^1,...,R^u$ contains $\bar{z}_j$. Taking the limit $n\to\infty, \frac{n}{m}\to 0$, a valid tiling is only possible if $f^j(\bar{z}_0)\in R$ for all $j\in\Z$. $\bar{z}_0$ is then called an \emph{immortal point} of the \emph{dynamical system} given by $M$, $c$ and $R$.

The breakthrough in \cite{kari1,kari2} was to realize that, conversely, there always exists a \emph{finite} alphabet $A$ such that any immortal point of the system $(M,c, R)$ can be represented with a valid tiling of $A$. 

Let $\bar{v}\in\R^2$, and, for any $k\in\Z$, define $A_k(\bar{v})\equiv \lfloor k\bar{v}\rfloor$, where the floor is taken for each coordinate of $\bar{v}$. Now, denote $B_k(\bar{v})\equiv A_k(\bar{v})-A_{k-1}(\bar{v})$. It can be shown that, if $\bar{v}\in R^j$, then $B_k(\bar{v})\in U^j$. Similarly, $\frac{1}{N}\sum_{k=1}^N B_{k+h}(\bar{v})=\bar{v}\pm O(1/N)$.

Now, suppose that, for any $i=1,...,u$, and any immortal point $\bar{v}\in R^i$, $A$ contains the set $A'$ of all tiles $T_k(\bar{v})$ with region $i$ and top, bottom, left and right vectors of the form:

\begin{align}
t_k(\bar{v})&=B_k(\bar{v}),\nonumber\\
b_k(\bar{v})&=B_k(f(\bar{v})),\nonumber\\
l_k(\bar{v})&=f(A_{k-1}(\bar{v}))-A_{k-1}(f(\bar{v}))+(k-1)c,\nonumber\\
r_k(\bar{v})&=f(A_k(\bar{v}))-A_{k}(f(\bar{v}))+kc,
\end{align}
\noindent for all $k\in\Z$. It can be verified that these tiles satisfy the conditions $f(t)+l=b+r$, $t\in U^i$, $b\in \cup_{j=1}^uU^j$.

Since $r_k(\bar{v})=l_{k+1}(\bar{v})$, the tiles $T_{k}(\bar{v}),T_{k+1}(\bar{v}),T_{k+2}(\bar{v}),...,T_{k+m}(\bar{v})$ form a valid row for any $m$. We can extend this row below by placing another row of the form $T_{k}(f(\bar{v})),T_{k+1}(f(\bar{v})),T_{k+2}(f(\bar{v})),...$, with $M\bar{v}\in R^{i'}$. Iterating, we can tile the whole plane in this fashion, provided that $\bar{v}$ is an immortal point. From the properties of $B_k(\bar{v})$, the averages of the inputs of each row $j$ tend to $f^j(\bar{v})$ in the limit $n\to\infty$, $\frac{n}{m}\to 0$.

It rests to see that $A'$ is a finite set. The top and bottom vectors of each tile belong to $\cup_{i=1}^uU^i$, and so they can only take finitely many values. As for the left or right vectors of the tile, they must belong to the set $L'=\{r_k(\bar{v}):\bar{v}\in\R^2,k\in\Z\}$. Let $l\in L'$ and denote by $l_s$ its $s^{th}$ component. Using the relations $x-1\leq\lfloor x\rfloor\leq x$ repeatedly, we have that  

\be
\mu^{-}_s:=-\sum_{j^{+}}M_{s,j}+c_s\leq l_s\leq 1-\sum_{j^{-}}M_{s,j}+c_s=:\mu^{+}_s
\ee
\noindent where $j^{+}$ ($j^{-}$) ranges over all those $j$ such that $M_{s,j}>0$ ($M_{s,j}<0$).

Call $m_s$ the common denominator $m_s$ of the rational numbers $\{M_{s,j}\}_j, c_s$ (remember that $M, c$ are assumed to be rational). From the definition of $L'$, it follows that $l\in L'$ must satisfy $\left(\begin{array}{cc}m_1&0\\0&m_2\end{array}\right)l\in\Z^2$. Define then the set $L=\{l:m_sl_s\in\Z, \mu^{-}_s\leq l_s\leq \mu^{+}_s, s=1,2\}$. Clearly $L$ is bounded and contains $L'$.

In order to generate $A$, we go through all regions $i,j=1,...,u$ and all combinations of top and bottom vectors $t\in U^i,b\in U^j$, and any left vector $l\in L$ and verify that the right vector $r=f(t)+l-b$ also belongs to $L$. If it does, then we add the tile $(i,t,l,r)$ to the definition of $A$. That way, we end up with a finite alphabet that, via the above tiling rules, can describe all immortal orbits of the dynamical system $(M,c,R)$. This is a simplification of the construction proposed in \cite{kari2} to simulate a dynamical system where the transformation $f$ may depend on the region $i$, i.e., $z\to M^iz+c^i$.

To prove our result, we take $M, c$ to be

\be
M=\left(\begin{array}{cc}\frac{4}{5}&-\frac{3}{5}\\\frac{3}{5}&\frac{4}{5}\end{array}\right), c=\left(\begin{array}{c}\frac{1}{5}\\\frac{1}{5}\end{array}\right),
\ee
\noindent and consider the regions $R^1, R^2$ corresponding to the unit squares $[0,-1]\times [0,1]$, $[0,1]\times [0,1]$, respectively. Using the above construction, we find that we can simulate the dynamical system defined by $(M, c, R)$ via an alphabet of $2947$ tiles.

The action of $f$ over a vector $\bar{z}$ can be rewritten $f(z)= M(z-\hat{c})+\hat{c}$, with $\hat{c}=(-1/5,2/5)$, where $M$ implements a counter-clockwise rotation by an angle $\phi_0=\arccos\left(\frac{4}{5}\right)$. Sequential applications of $f$ over an initial vector $z_0$ have thus the effect of rotating the vector an angle $\phi_0$ with respect to the point $\hat{c}$, see Figure \ref{giros}.

\begin{figure}
  \centering
  \includegraphics[width=0.6\textwidth]{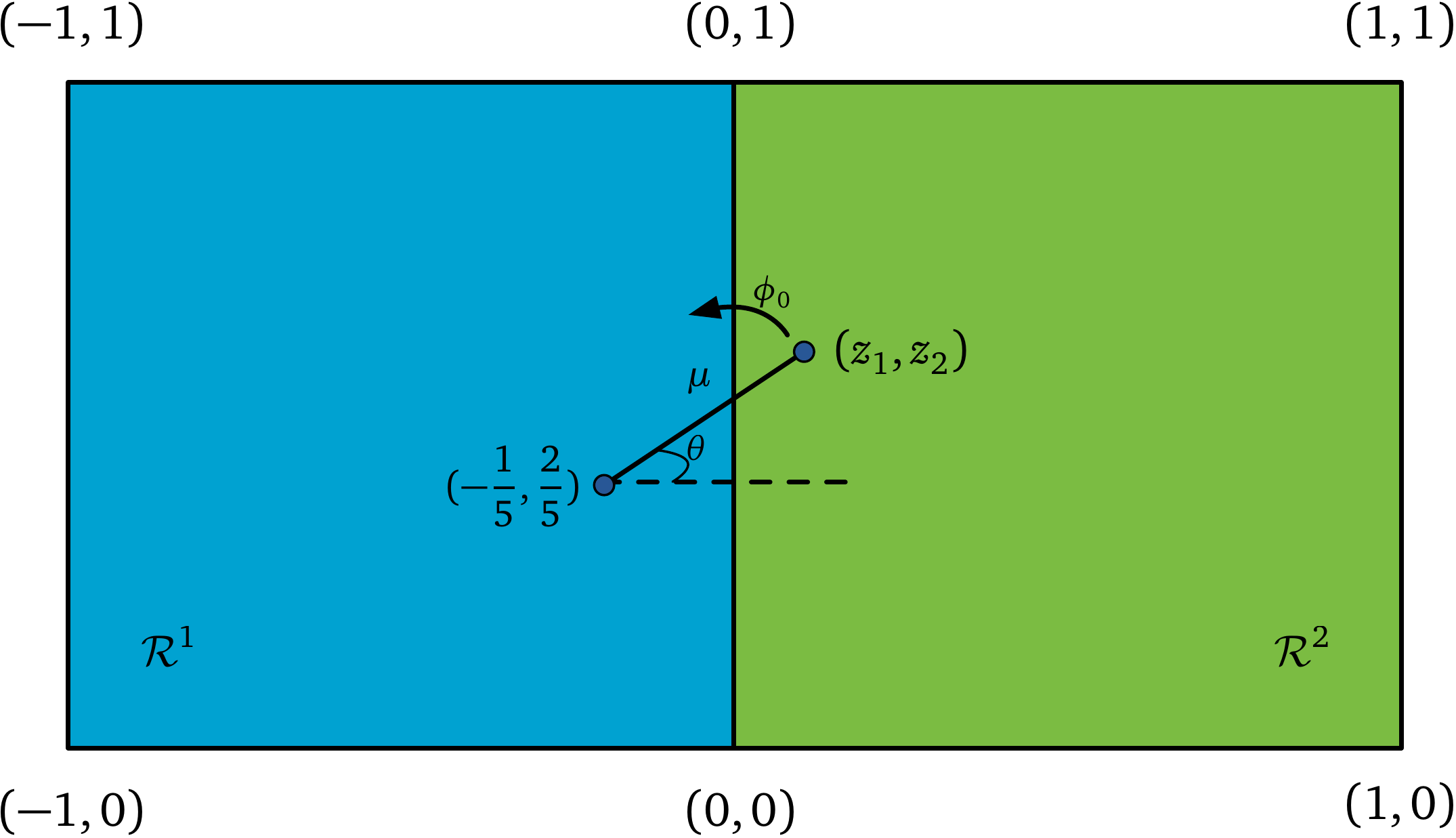}
  \caption{\textbf{The action of $f$.} The effect of the affine transformation $f$ over an arbitrary vector $z$ is to rotate it by an angle $\phi_0$ with respect to the point $\hat{c}=(-1/5,2/5)$.}
  \label{giros}
\end{figure}

As we will see later, $\phi_0$ is an irrational multiple of $2\pi$. Hence, by applying $M$ sequentially, we can induce a rotation arbitrarily close to any angle $\theta$. It follows that a point $\bar{z}$ is immortal in the dynamical system given by $(M,c,R)$ iff $\|\bar{z}-\hat{c}\|_2\leq \frac{2}{5}$.

We are interested in the following extensive quantities:
\begin{align}
&\Omega\equiv\frac{1}{(2m+1)(2n+1)}\sum_{x=-m}^{m}\sum_{y=-n}^{n}\left(t_1(x,y)+\frac{1}{5}\right)\delta_{i(x,y),2},\label{witness1}\\
&H\equiv\frac{1}{(2m+1)(2n+1)}\sum_{x=-m}^{m}\sum_{y=-n}^{n}\delta_{i(x,y),2},
\label{witness2}
\end{align}
where $i(x,y),t(x,y),b(x,y),l(x,y),r(x,y)$ denote the parameters specifying the tile at position $(x,y)$; $t_1$ is the first coordinate of the top vector $t$; and $\delta$, the Kronecker delta. We will show that, if $n\to\infty$ and $\frac{n}{m}\to 0$, the set of feasible $(\Omega,H)$ is parametrized by the curve 

\be
\left\{\left(\frac{\mu}{\pi}\sqrt{1-\left(\frac{1}{5\mu}\right)^2}, \frac{1}{\pi}\arccos\left(\frac{1}{5\mu}\right)\right): \frac{1}{5}\leq \mu\leq \frac{2}{5}\right\}.
\label{curvas}
\ee

Intuitively, the vector $z$ defined at the top side of each row of the tiling is turning by an amount $\phi_0$ from row to row. The witness (\ref{witness1}) corresponds to the average of $\bar{z}-\hat{c}$'s first coordinate in the region $R^2$. That is,
\be
W=\lim_{N\to\infty}\frac{1}{2N+1}\sum_{k=-N}^N\mu\cos(\varphi+k\phi_0)\chi_{[-\arccos\left(\frac{1}{5\mu}\right),\arccos\left(\frac{1}{5\mu}\right)]}(\varphi+k\phi_0),
\ee

\noindent where $\varphi$ is the angle of $\bar{z}$ at row $j=0$ with respect to the $\hat{x}$ axis; $\mu=\|z-\hat{c}\|_2$ and $\chi_O(p)$ is the characteristic function that equals $1$ when $p\in O$ and $0$ otherwise. Since $\phi_0$ is not congruent, we expect the above expression to converge to

\be
\frac{1}{2\pi}\int_{-\arccos\left(\frac{1}{5\mu}\right)}^{\arccos\left(\frac{1}{5\mu}\right)}d\theta \mu\cos(\theta)=\frac{\mu}{\pi}\sqrt{1-\left(\frac{1}{5\mu}\right)^2},
\label{goal1}
\ee

\noindent for $\frac{1}{5}\leq\mu\leq\frac{2}{5}$ and $0$ otherwise. Analogously, the second witness (\ref{witness2}) measures the presence of vector $z$ in $R^2$, and so it should converge to

\be
\frac{1}{2\pi}\int_{-\arccos\left(\frac{1}{5\mu}\right)}^{\arccos\left(\frac{1}{5\mu}\right)}d\theta =\frac{1}{\pi}\arccos\left(\frac{1}{5\mu}\right),
\label{goal2}
\ee
\noindent for $\frac{1}{5}\leq\mu\leq\frac{2}{5}$ and $0$ otherwise. The coordinates of an arbitrary feasible point $(\Omega,H)$ will thus belong to the trajectory (\ref{curvas}).

To justify eqs. (\ref{goal1}), (\ref{goal2}), though, we need to prove a result that relates the integration of a function with the sampling over $(\varphi+k\phi_0)_k$.

\begin{lemma}
Let $\phi_0$ be such that $m\phi_0\not=0\mbox{ }(\mbox{mod }2\pi)$ for all $m$, and let $g(\phi)$ be any piecewise continuous bounded function in $\phi\in[-\pi,\pi]$. Then, for any $\varphi\in[-\pi,\pi]$,

\be
g^\star:=\lim_{N\to\infty}\frac{1}{N}\sum_{k=1}^Ng(\varphi+k\phi_0)=\frac{1}{2\pi}\int_{-\pi}^{\pi} d\theta g(\theta).
\label{integral}
\ee
\end{lemma}
\begin{proof}
Note that, for any $m\in\Z$,

\begin{align}
\lim_{N\to\infty}\frac{1}{N}\sum_{k=1}^Ne^{m(\varphi+k\phi_0)}=&0,\mbox{ for } m\not=0,\nonumber\\
&1\mbox{ for } m=0.
\label{Fourier}
\end{align}

\noindent Let $g(\theta)$ be a periodic, piecewise \emph{differentiable} bounded function. Then it converges uniformly under a Fourier expansion, i.e., for any $\epsilon>0$, there exists $R$ such that $|g(\theta)-g_R(\theta)|<\epsilon$ for all $\theta\in[-\pi,\pi]$, where $g_R(\theta)=\sum_{m=-R}^R c_m e^{m\theta}$ and $\{c_m\}_m$ are the Fourier coefficients of $g(\theta)$.

Now, from eq. (\ref{Fourier}), we have that

\be
g^\star_R=c_0=\frac{1}{2\pi}\int_{-\pi}^{\pi}d\theta g(\theta), \forall R.
\ee
\noindent Since the $\star$ process is, in fact, an average, $|g^\star_R-g^\star|<\epsilon$, and we conclude that eq. (\ref{integral}) holds for all piece-wise derivable bounded functions. The general result can be obtained by noticing that, for any piece-wise continuous function $g$, there exist two sequences of piece-wise derivable functions $(g^{u}_n)_n$, $(g^{d}_n)_n$ such that $g^{d}_n(\theta)\leq g(\theta)\leq g^{u}_n(\theta)$ for all $n,\theta$ and $\lim_{n\to\infty}\int_{-\pi}^\pi d\theta g^d_n(\theta)=\lim_{n\to\infty}\int_{-\pi}^\pi d\theta g^u_n(\theta)$.
\end{proof}

To apply the above lemma, we still need to show that $\phi_0$ is congruent. This follows from $\sin(\phi_0)=\frac{3}{5}$ and Niven's theorem \cite[pp.41]{niven}, that states that the only angles $0\leq\phi\leq \pi/2$ with $\phi=\frac{m}{n}\pi$ and rational sine are $0$, $\frac{\pi}{6}$ and $\frac{\pi}{2}$.

We have just shown that, in the limits $n\to\infty$, $\frac{n}{m}\to 0$, the vector of feasible values $(\Omega,H)$ in an $n\times m$ valid tiling is parametrized by eq. (\ref{curvas}). What does this have to do with TI marginals? Consider a random variable taking values in the tile set $A$, and a TI marginal $(P_h, P_v)$ satisfying the constraint:

\be
\sum_{(a,b)\in\T_h}P_h(a,b)=\sum_{(a,b)\in\T_v}P_v(a,b)=1.
\label{constance}
\ee

\noindent Consider also the linear functionals given by

\begin{align}
&\omega(P_h,P_v)\equiv\sum_{a\in A}P(a)\left(t_1(a)+\frac{1}{5}\right)\delta_{i(a),2},\nonumber\\
&\eta(P_h,P_v)\equiv\sum_{a\in A}P(a)\delta_{i(a),2},
\label{witnessP}
\end{align}
\noindent where $t_1(a)\in\{0,\pm 1\}$ $(i(a)\in \{1,2\})$ denotes the 2nd coordinate of the top side (the region) of tile $a$.

By definition, for any $n, m$ there exists a LTI distribution $P_{n\times m}$ with nearest-neighbor marginals $(P_h, P_v)$. $P_{n\times m}$ can be seen as a convex combination of tilings $a_{n\times m}^k$ with weight $p_k$ of the set $n\times m$. Moreover, due to condition (\ref{constance}), all of them must be valid tilings. Hence,

\be
\omega(P_h,P_v)=\sum_kp_k \Omega(a_{n\times m}^k), \eta(P_h,P_v)=\sum_kp_k H(a_{n\times m}^k).
\ee

\noindent In the limit $n\to\infty$, $\frac{n}{m}\to 0$, the point $(\omega,\eta)$ belongs to the convex hull ${\cal S}$ of the curve (\ref{curvas}). It can be verified that $\frac{d^2\eta}{d\omega^2}\leq 0$, i.e., the curve is concave. Hence the boundary of ${\cal S}$ is given by curve (\ref{curvas}) and the segment joining its start and end points, see Figure \ref{non_polytope}.

\begin{figure}
  \centering
  \includegraphics[width=0.6\textwidth]{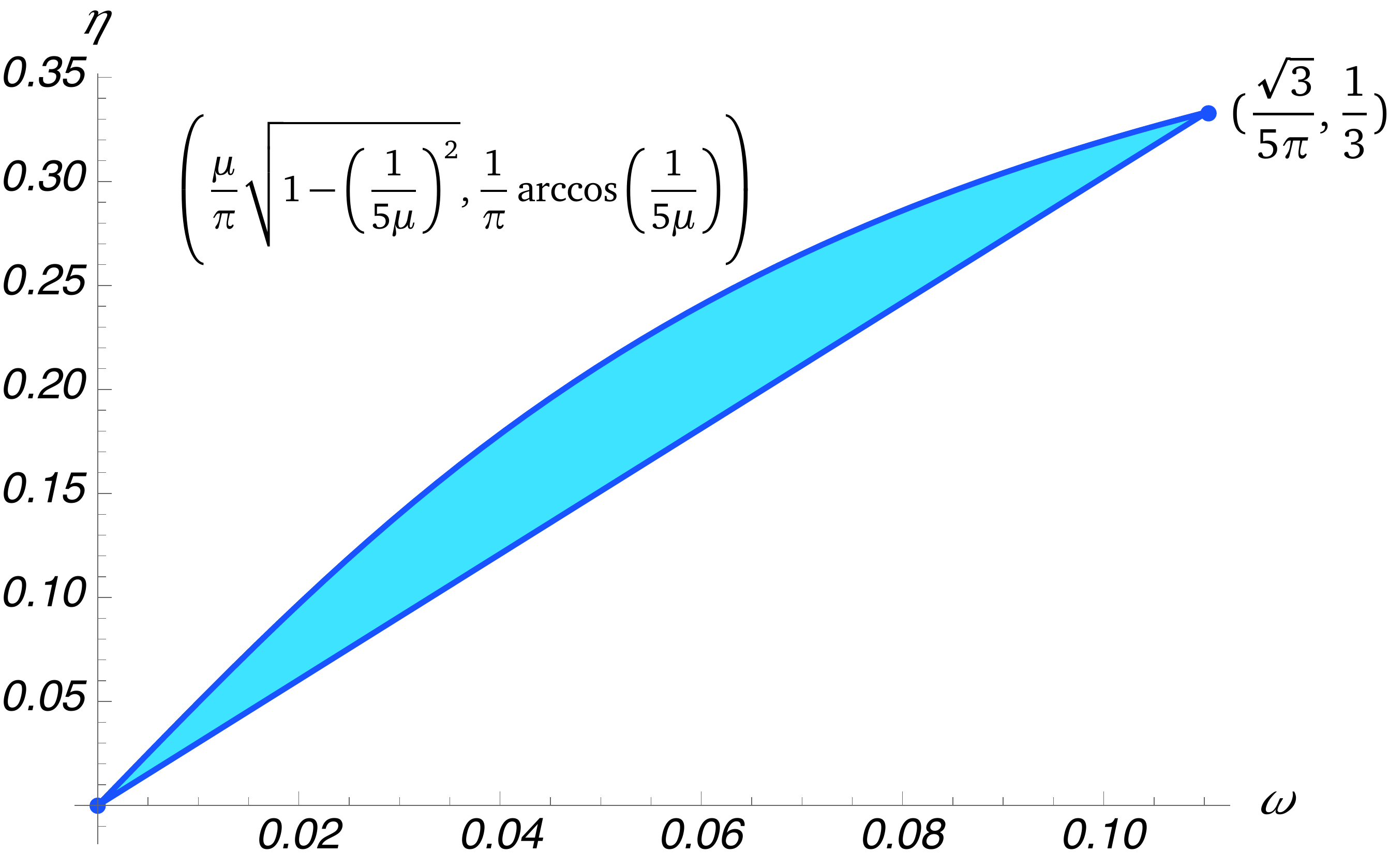}
  \caption{\textbf{Accessible values for the parameters $\omega$ and $\eta$ are given by the shaded region.}}
  \label{non_polytope}
\end{figure}

It rests to show that any point of the shaded region in Figure \ref{non_polytope} is achievable by a TI marginal. Since the set of TI marginals is convex, it is enough to see that one can achieve the extreme points of the set ${\cal S}$. Let $\frac{1}{5}\leq\mu\leq \frac{2}{5}$ and take any valid tiling $a_{Z^2}$ of the plane describing a vector $z\in\R^2$ with $\|z-\hat{c}\|_2=\mu$. Further take any increasing sequence of rectangles $m^{(k)}\times n^{(k)}$, with $\lim_{k\to\infty}n^{(k)}/m^{(k)}=0$. Symmetrizing the deterministic distribution $a_{m^{(k)}\times n^{(k)}}$, we obtain a sequence of TI marginals $(P^{(k)}_h,P^{(k)}_h)_k$ which violate (\ref{constance}) by an amount $O(1/n^{(k)})$. By Weiestrass' theorem, there exists a subsequence of this sequence that converges to a pair of distributions, call them $(P_h,P_v)$. By closure of the set of TI marginals, $(P_h,P_v)$ admits a 2D TI extension. Moreover, it satisfies (\ref{constance}) and $\omega,\eta$ are are given by eq. (\ref{curvas}).

As it is clear from Figure \ref{non_polytope}, the upper piece of ${\cal S}$'s boundary is smoothly curved. In other words: ${\cal S}$ is not a polytope, and so neither is the set of nearest-neighbor TI marginals. In order to show that it is also not a semi-algebraic set, see eq.  (\ref{basic_closed}), we will assume that it is and prove the result by contradiction. 

${\cal S}$ is the result of intersecting the set of TI marginals with a number of planes followed by a projection on the variables $\omega,\eta$: from our hypothesis, it follows that it is also a basic closed semi-algebraic set. Adding the variables $\nu,\mu$ and the relations $\mu^2-\frac{1}{25}-\pi^2\omega^2=0$, $\mu\geq \frac{1}{5}$, $\nu\geq 0$, $\mu\nu-\frac{1}{5}=0$, we have a new closed semi-algebraic set for the variables $(\omega,\eta,\mu,\nu)$. Its projection onto $(\nu,\eta)$ gives rise to a new semi-algebraic set ${\cal S}'$, with its boundary containing the curve $\{(\nu,\frac{\arccos(\nu)}{\pi}):\nu\in[\frac{1}{2},1]\}$. 

By the Tarski-Seidenberg projection theorem \cite{tarski,Bochnak1998}, ${\cal S}'$ is determined by a number of polynomials $\{G_i\}_{i=1}^u$, such that

\begin{align}
&(\nu,\eta)\in\tilde{{\cal S}'} \mbox{ iff }\nonumber\\
&G_i(\omega, \eta)\geq 0, i=1,...,u.
\label{basic_closed2}
\end{align}

Now, for any point $\hat{s}$ on the boundary of $\tilde{{\cal S}}'$ at least one of these polynomials must be null; otherwise, we could perturb $\hat{s}$ in any direction and the polynomial inequalities (\ref{basic_closed}) would still hold. It follows that any point $\hat{s}$ on the boundary of $\tilde{{\cal S}}'$ satisfies $g(\hat{s})=0$, where $g(\hat{s})\equiv\prod_{i=1}^uG_i(\hat{s})$. In particular, we have that

\be
g\left(\nu, \frac{\arccos(\nu)}{\pi}\right)=0,
\ee
\noindent for $\nu\in [1/2,1]$. This contradicts the fact that the inverse of cosine is a transcendental function.
\end{proof}

\section{Conclusion}

In this work, we have studied the problem of deciding whether a number of distributions correspond to the marginals of a 2D TI system, what we called the $\MARGINAL$ problem. We found that this problem is exactly solvable in scenarios of low local dimension and nearest or next-to-nearest neighbor statistics. We also showed that a natural variant of the problem, where we also demand symmetry under reflection, is solvable for all local dimensions. For other scenarios, we proposed a general algorithm to approximately solve the $\MARGINAL$ problem, as well as its dual, the $\ENERGY$ problem, where the goal is to minimize a linear functional of TI marginals.

We also proved several no-go theorems concerning these two problems. We showed that the $\ENERGY$ problem is undecidable in general, so we cannot expect to identify the sets of TI marginals exactly for arbitrary $d$. We find that for $d$ high enough, those sets are neither real polytopes nor semi-algebraic sets. Our techniques to prove negative results relied on a correspondence, proposed by Kari \cite{kari1,kari2}, between aperiodic tilings and immortal points of dynamical systems. Perhaps because of this, our upper bounds on the minimal dimension over which the set of TI marginals ceases to admit a simple description seem very poor. It is an open question how to lower those bounds. Could it be that, already for dimensions of order $10$, we can experience the transition from a rational polytope to a convex object where parts of the boundary are smoothly curved? And, could it be that, for dimensions small enough, the description of the sets ceases to be a polytope but nonetheless admits a practical description via semidefinite programming?

Finally, we hope our methods and results can be applied to the study of other thermodynamical quantities or out-of-equilibrium systems. An immediate follow-up to our work would be to relate TI marginals to the maximum entropy per site of the whole lattice configuration from which they originate. That way, we would be able to compute interesting thermodynamical quantities of TI systems, such as the free energy, at non-zero temperature. Such maximum entropy extensions of TI marginals have been studied on the 1D Euclidean lattice and the Bethe lattice~\cite{Goldstein2017}, but nothing is known when the lattice is 2D Euclidean. In this regard, our results on the uniqueness of the extendibility of the extreme points of $d=2$ TI marginals suggest that a full characterization of the set of achievable TI marginals plus maximum entropy in that scenario is on the horizon.

\ack{
M.N. acknowledges interesting and useful discussions with David P\'erez-Garc\'ia. The authors would like to thank Aernout van Enter for pointing out several related works.}

\dataccess{This work does not have any experimental data.}

\aucontribute{Both authors contributed equally to this work. Both authors gave final approval for publication.}

\competing{We have no competing interests.}

\funding{This work was supported by the FQXi grant ``The physics of events''.}

\bibliographystyle{unsrtnat}
\bibliography{biblio}

\end{document}